\documentclass[11pt]{article}
\usepackage{amsmath}
\usepackage{float}
\usepackage{theorem}
\usepackage{graphicx}
\usepackage{amsfonts}
\usepackage{amssymb}
\theorembodyfont{\upshape}
\newtheorem{theorem}{Theorem}

\newtheorem{assumption}{Assumption}

\newtheorem{corollary}{Corollary}

\newtheorem{definition}{Definition}
\newtheorem{example}{Example}

\newtheorem{notation}{Notation}

\newtheorem{remark}{Remark}

\newenvironment{proof}[1][Proof]{\textbf{#1.} }{\ \rule{0.5em}{0.5em}}
\advance\topmargin by -0.75truein
\advance\textheight by 1.25truein
\begin{document}

\begin{center}
{\LARGE On modeling hard combinatorial optimization problems as linear
programs: Refutations of the ``unconditional impossibility'' claims}
{\Large \medskip\medskip}

Moustapha Diaby

OPIM Department; University of Connecticut; Storrs, CT 06268\newline
moustapha.diaby@uconn.edu\medskip{\Large \medskip}

Mark H. Karwan

Department of Industrial and Systems Engineering; SUNY at Buffalo; Amherst, NY
14260\newline mkarwan@buffalo.edu\medskip{\Large \medskip}

Lei Sun

Department of Industrial and Systems Engineering; SUNY at Buffalo; Amherst, NY
14260\newline leisun@buffalo.edu\medskip{\Large \medskip} \ 
\end{center}

\textsl{Abstract.}{\small \ There has been a series of developments in the
recent literature (by essentially a same ``circle'' of authors) with the
absolute/unconditioned (implicit or explicit) claim that there exists no
abstraction of an \textit{NP-Complete} combinatorial optimization problem in
which the defining combinatorial configurations (such as ``tours'' in the case
of the traveling salesman problem (TSP) for example) can be modeled by a
\textit{polynomial-sized }system of linear constraints. The purpose of this
paper is to provide general as well as specific refutations for these recent
claims. \bigskip}

\textsl{Keywords:}\textbf{\ }{\small Linear Programming; Combinatorial
Optimization; Computational Complexity; Traveling Salesman Problem; TSP;
``}${\small P}${\small \ vs. }${\small NP}${\small .''\bigskip}

\section{Introduction}
\label{Introduction_Section}

\noindent Combinatorial optimization problems (COPs) occupy a central place in
Operations Research (OR) and in Mathematics in general. Part of the reason for
this is the wide practical applicability of COPs, as virtually every
scheduling, sequencing, or routing problem that arises in industry involves a
COP. COPs are also central for the field of Theoretical Computer Science (TCS)
in particlar, because COPs serve as a foundation for the study of
\textit{computatibity,} which in turn, has to do with the very foundations of
Mathematics in general, and specifically the issues of the generalizability of
the \textit{axiomatic approach} (Hilbert (1898)) and the so-called
\textit{incompleteness theorem }(G\"{o}del (1931)). Since the beginnings of
their field, Operations Researchers have approached COPs from an
``Engineering'' perspective, with the focus on developing/``engineering''
methods aimed at obtaining ideal or satisfactory solutions to the practical
problems they model. The perspectives in TCS and Mathematics in general, on
the other hand, have been much more abstract, with a focus on what is/is not
possible in terms of \textit{computability} (G\"{o}del (1931)).

A convergence between the OR and TCS perspectives was brought about by the
landmark invention of the theory of \textit{NP-Completeness} (Cook (1971); see
also Garey and Johnson (1979)) on one hand, and the landmark developments of
\textit{interior-point methods} (Katchyian (1979); Karmarkar (1984)) on the
other hand. \textit{NP-Completeness} delineated the so-called ``$NP$'' class
of problems, and within it, the subclass ``$P$'' of problems that are easily
``computable'' according to some theoretical measure, thereby essentially
reducing the question of \textit{computability} to the question of whether the
``$P$'' and ``$NP$'' classes are equal. With \textit{interior-point methods}
came the discovery that the ``Engineering'' problems called ``Linear Programs
(LPs)'' (see Bazaraa \textit{et al}. (2006)) fall within the ``$P$'' class of
problems, and this reduced the \textit{computability} question for engineers
to that of whether or not any of the representatives of the ``hard''
$NP$-class of problems (i.e., the so-called \textit{NP-Complete} problems)
could be modeled/``engineered'' as a \textit{polynomial-sized} LP. A COP is
said to have been ``modeled as an LP'' if it has been abstracted into an LP
which has integral extreme points corresponding to the combinatorial
configurations being decided upon. An LP model is said to be of
\textit{polynomial size} if its numbers of variables and constraints are
respectively bounded by polynomial functions of a measure of the \textit{size}
of the input data for the problem being modeled.

As far as we know, the first \textit{polynomial-sized} LP models to be
proposed for an \textit{NP-Complete} problem (specifically, the traveling
salesman problem (TSP); see Lawler \textit{et al}. (1985)) are those of Swart
(1986; 1987). Not having/having had access to Swart's papers, we do not know
of any of the details of his models. However, the concensus of the research
communities is that their validities were refuted by Yannakakis (1991). The
essence of the Yannakakis (1991) developments is the showing that for the TSP,
the polytope which is stated in terms of its extreme points in the space of
the ``natural'' city-to-city variables (i.e., the so-called ``\textit{The} TSP
Polytope''; see Lawler \textit{et al}. (1985)) does not have a
\textit{polynomial-sized extended formulation} ($EF$) which is
\textit{symmetric }(see Yannakakis (1991))\textit{.} Roughly, a system of
linear constraints (say, ``System $A$'') is an $EF$ (or an
``\textit{extension}'') of another system of linear constraints (say, ``System
$B$'') if the polytopes they respectively induce in the space of the ``System
$B$'' variables coincide (see Yannakakis (1991) or Balas (2005), among
others). If ``System $A$'' is an $EF$ of ``System $B$,'' the polytope induced
by ``System $A$'' is said to be an $EF$ of the polytope induced by ``System
$B$,'' and also, ``System $B$'' (or the polytope it induces) is said to be the
\textit{projection of ``System }$A$\textit{'' }(or the polytope it induces)
\textit{on the space of ``System }$B$.\textit{''} Clearly, a given system of
equations cannot induce a polytope in the space of variables which are not
part of that system. Hence, as developed in Diaby and Karwan (2016; 2017), the
addition of redundant variables (and possibly, constraints also) to a system
of constraints leads to \textit{degenerate}/non-meaningful $EF$ relationships
with respect to the task of making inferences about model \textit{sizes}.

Yannakakis' (1991) work is carefully scoped. However, it fails to make
appropriate explicit exceptions for cases involving zero-matrices in its
analyses, thereby failing to explicitly distinguish between
\textit{degenerate}/non-meaningful $EF$s from meaningful ones with respect to
the puposes of those analyses. We argue that the lack of recognition of this
distinction is what has led to the increasingly over-scoped claims in the
recent $EF$ literature which attempts to build on Yannakakis' (1991) approach.
Specifically, the claim in Fiorini \textit{et al}. (2011; 2012) is that the
natural descriptions of the traveling salesman and stable set polytopes do not
have \textit{polynomial-sized EF}s, regarless of the \textit{symmetry}
condition of Yannakakis (1991). Theoretical and numerical refutations of these
Fiorini \textit{et al}. (2011; 2012) developments are provided in Diaby and
Karwan (2016; 2017). In Avis and Tiwary (2015), Bri\"{e}t \textit{et al}.
(2015), Braun \textit{et al}. (2015), Fiorini \textit{et al}. (2015) (which is
essentially the ``journal version'' of Fiorini \textit{et al.} (2011; 2012)),
Lee \textit{et al}. (2017), and Averkov \textit{et al}. (2018) respectively,
the claim is the absolute/unconditioned (implicit or explicit) statement that
there exists no abstraction of an \textit{NP-Complete} COP in which the
defining combinatorial configurations (such as ``tours'' in the case of the
TSP for example) can be modeled by a \textit{polynomial-sized }system of
linear constraints. The purpose of this paper is to provide general as well as
specific refutations for these recent claims.

The plan of the paper is as follows. We will discuss general-level refutations
in section \ref{Gen'l_Refut_Section}, where we will provide a
\textit{polynomial-sized} system of linear constraints which correctly
abstracts TSP tours, and also derive conditions for the existence of an affine
map (which a linear map is a special case of) between disjoint sets of
variables with no implication nor need for \textit{extended formulations}
($EF$) relationships. Specific refutations will be discussed in section
\ref{Specific_Refutations_Section}, using the Fiorini \textit{et al}. (2015)
developments. Finally, some concluding remarks will be offered in section
\ref{Conclusion_Section}.

\section{General refutations}
\label{Gen'l_Refut_Section}

The fundamental presumption upon which the developments in the ``unconditional
impossibility'' papers (Avis and Tiwary (2015), Bri\"{e}t \textit{et al}.
(2015), Braun \textit{et al}. (2015), Fiorini \textit{et al}. (2015), Lee
\textit{et al}. (2017), and Averkov \textit{et al}. (2018)) rest is that the
\textit{size} of the description of a polytope in the space of its variables
can be inferred from that of another polytope stated in a disjoint space of
variables. The objective of this section is to provide general, direct
refutations of this misconception, and using the TSP, of the over-reaching
``impossibility claim'' itself. We will first recall the definition of the
standard TSP polytope (i.e., the so-called ``\textit{The} TSP Polytpe'').
Then, we will present an alternate, \textit{polynomial-sized} system of linear
constraints which correctly abstracts TSP tours. Finally, we will show that
the existence of an affine map between disjoint sets of variables (as was
brought to our attention by Kaibel \textit{et al}. (2013)) is not a sufficient
condition for the existence of $EF$ relationships from which valid comparisons
between the descriptions of polytopes stated in those (disjoint) spaces can be
made. For the discussions about the TSP\ polytopes, we will use the following conventions.\medskip

\begin{assumption}
\ \ \bigskip

\noindent We assume without loss of generality that:\smallskip

\begin{enumerate}
\item A city designated as ``$0$'' is the beginning and ending point of all travels;

\item The TSP tours have been ordered, with the $k^{th}$ one designated by
$T_{k}$ ($k\in\{1,\ldots,(n-1)!\}$).\medskip
\end{enumerate}
\end{assumption}

\begin{notation}
\label{Gen'l_Refutations_Notation} \ \ \ \ 

\begin{enumerate}
\item $n:$ Number of cities;

\item $\Omega:=\{0,\ldots,n-1\}$ \ \ (Index set of the cities);

\item $\mathcal{A}:=\left\{  \left(  i,j\right)  \in\Omega^{2}:i\neq
j\right\}  $ \ \ (Set of arcs of the TSP (city-to-city) graph; Set of possible
TSP ``travel legs'');

\item $m:=n-1;$

\item $M:=\Omega\backslash\{0\}$ \ \ (Set of cities to visit when city ``$0$''
is considered the starting and ending point of the travels);

\item $S:=\{1,\ldots,m\}$ \ \ (Index set for the ``times-of-visit'' for the
cities in $M$; see Picard and Queyranne (1978));

\item $\forall(i,j)\in\mathcal{A},$ $x_{ij}$ $:$ Variable indicating whether
city $i$ is visited immediately before city $j$ ($x_{ij}=1$)$,$ or not
($x_{ij}=0$);

\item $\forall(i,s)\in(M,S),$ $w_{is}:$ Variable indicating whether city $i$
is visited at ``time'' $s$ ($w_{is}=1$)$,$ or not ($w_{is}=0$);

\item $Conv(A):$ Convex hull of $A$.\medskip
\end{enumerate}
\end{notation}

\begin{definition}
[Standard TSP polytope: ``\textit{The} TSP Polytope'']
\label{"The_TSP_Polytope"_Dfn} \ \ \ 

\begin{description}
\item $\forall$ $F\subseteq$ $\mathcal{A}:F\neq\varnothing,$ let
$x^{F}:=\left\{  x\in\{0,1\}^{n(n-1)}:x_{ij}=1\text{ iff }(i,j)\in F\right\}
$. The standard TSP polytope

\item (i.e., ``\textit{The} TSP Polytope'') is defined as $Conv\left(
\left\{  x^{T_{k}},\text{ }\left(  k=1,\ldots,(n-1)!\right)  \right\}
\right)  .$
\end{description}
\end{definition}

\subsection{There exists an alternate ``TSP
polytope''}
\label{Alt_TSP_Polytope_SubSection}

In this section, we will present a \textit{polynomial-sized} system of linear
constraints which correctly abstracts TSP tours, and illustrate it with a
numerical example.

\begin{theorem}
\label{AP_Polytope_Thm}
The extreme points of \[AP:=\left\{  \mathbf{w}\in\mathbb{R}^{(n-1)^{2}}:\sum\limits_{s\in S}
w_{is}=1\text{ \ }\forall i\in M;\text{ \ }\sum\limits_{i\in M}w_{is}=1\text{
\ }\forall s\in S;\text{ \ }\mathbf{w}\geq\mathbf{0}\right\}\]
are in one-to-one correspondence with TSP tours which start and end at city ``$0$''.
\end{theorem}

\begin{proof}
\textit{Polytope} $AP$ is the standard Linear Assiggment Problem (or
\textit{Birkhoff}) polytope, and has therefore, integral extreme points (see
Birkhoff (1946); Burkhard \textit{et al}. (2009); among others). Moreover,
using the assumption that city $``0"$ is the starting and ending point of
travel, it is trivial to construct a unique TSP tour from a given extreme
point of $AP,$ and vice versa (i.e., it is trivial to construct a unique
extreme point of $AP$ from a given TSP tour), as shown below. \ 

Let $\mathbf{w}^{k}$ $(k\in\{1,\ldots,(n-1)!\})$ denote the $k^{th}$ extreme
point of $AP$, with corresponding set of assignments $C^{k}:=\left\{  \left(
(a_{p}^{k},p\right)  \in(M,S),\text{ }p=1,\ldots,n-1:\left(  \forall(p,q)\in
S^{2}:p\neq q,\text{ }a_{p}^{k}\neq a_{q}^{k}\right)  \right\}  .$ Then the
components of $\mathbf{w}^{k}$ are specified as follows:
\begin{equation}
\forall\left(  i,s\right)  \in\left(  M,S\right)  ,\text{ }w_{is}
^{k}=\left\{
\begin{array}
[c]{c}
1\text{ \ \ if }i=a_{s}^{k};\\
\\
0\text{ \ \ otherwise.}
\end{array}
\right.  \label{AP_Polytope_Proof(1)}
\end{equation}
The order of visits for the unique TSP tour $(T_{k})$ corresponding to
$\mathbf{w}^{k}$ is: $0\longrightarrow a_{1}^{k}\longrightarrow\ldots
\longrightarrow a_{n-1}^{k}\longrightarrow0.$

Conversely, let $T_{k}$ $(k\in\{1,\ldots,(n-1)!\})$ denote the $k^{th}$ TSP
tour, with order of visits specified as: $0\longrightarrow a_{1}
^{k}\longrightarrow\ldots\longrightarrow a_{n-1}^{k}\longrightarrow0$ (where
$a_{p}^{k}\in M$ for $p=1,\ldots,n-1$). The unique extreme point of $AP$,
$\mathbf{w}^{k}$, corresponding to $T_{k}$ is obtained by applying
(\ref{AP_Polytope_Proof(1)}) above. \ \ \medskip
\end{proof}

\begin{remark}
\label{Non-Applicability_Rmk} \ \ 

\begin{enumerate}
\item \label{NARmk(0.5)}$AP$ is distinct from the permutahedron (i.e., the
convex hull of all vectors that arise from permutations) of the TSP cities.

\item \label{NARmk(1)}$AP$ does not induce TSP tours \textit{per se} (i.e.,
Hamiltonian cycles of the TSP cities; see Lawler \textit{et al}. (1985)).

\item \label{NARmk(2)}It is not possible to make $AP$ induce TSP tours
\textit{per se} by adding a ``dummy'' city to the set of cities and using it
as the starting and ending point of the travels. The reason for this is that
it would not necessarily be possible to associate the $AP$ solution thus
obtained to a TSP tour, as the actual cities of the first and last
\textit{times-of-visit} may be different.

\item \label{NARmk(3)}The vertices of $AP$\ are linear assignment problem
(LAP; see Burkard \textit{et al}. (2009)) solutions, whereas the vertices of
``\textit{The} TSP Polytope'' model Hamiltonian cycles. Hence, $AP$ and
``\textit{The} TSP Polytope'' are mathematically-unrelated polytopes.\textit{\ }

\item \label{NArmk(3.5)}The association that can be made between $AP$ and
``\textit{The} TSP Polytope'' is \textit{cognitive} only.

\item \label{NARmk(4)}According to the \textit{Minkowski-Weyl Theorem}
(Minkowski (1910); Weyl (1935); see also Rockafellar (1997, pp.153-172)),
every polytope can be equivalently described as the intersection of
hyperplanes ($\mathcal{H}$-representation/external description) or as a convex
combination of (a finite number of) vertices ($\mathcal{V}$
-representation/internal description). ``\textit{The} TSP Polytope'' is easy
to state in terms of its $\mathcal{V}$-representation. However, no
\textit{polynomial-sized} $\mathcal{H}$-representation of it is known. On the
other hand, the $\mathcal{H}$-representation of $AP$ is well-known to be of
(low-degree) \textit{polynomial size} (see Burkard \textit{et al}. (2009)),
and it is trivial to state its $\mathcal{V}$-representation also.

\item \label{NARmk(5)}In order to model the TSP optimization problem as an
$LP$ in the space of the natural, \textit{travel-leg} $x_{ij}$ variables, an
$\mathcal{H}$-representation of ``\textit{The} TSP Polytope''\ must be
developed. This task has thwarted all efforts so far. In order to model the
TSP optimization problem as an $LP$ in the space of the LAP,
\textit{travel-time} $w_{ir}$ variables, a linear function which correctly
captures TSP tours costs must be developed. Hence, mathematical developments
focused on what is/is not possible to do for the linear system modeling of the
``\textit{The} TSP Polytope'' would have no pertinence in abstractions which
are based on $AP$ and do not require the natural $x_{ij}$ variables. Examples
of such abstractions are given in Diaby (2007), and Diaby and Karwan (2016).
\ \smallskip$\square$
\end{enumerate}
\end{remark}

The results below follow directly from the discussions above.

\begin{corollary}
\label{AP_Polytope_Coroll1}\ \ \ \ 

\begin{enumerate}
\item $AP$ is an alternate ``TSP polytope'' from the standard TSP polytope
(i.e., the so-called ``\textit{The} TSP Polytope'').

\item $AP$ is a refutation of the claim in the recent literature (Avis and
Tiwary (2015), Bri\"{e}t \textit{et al}. (2015), Braun \textit{et al}. (2015),
Fiorini \textit{et al}. (2015), Lee \textit{et al}. (2017), and Averkov
\textit{et al}. (2018)) that the combinatorial configurations which define an
\textit{NP-Complete} problem cannot be abstracted into a
\textit{polynomial-sized} system of linear constraints.\ 
\end{enumerate}
\end{corollary}

\begin{definition}
\label{TL&TT_TSP_Polytope} \ 

\begin{enumerate}
\item The standard TSP polytope may be referred to as the \textit{travel-legs
(TL) TSP polytope}, and its extreme points may be referred to as
\textit{travel-legs (TL) TSP tours}.

\item $AP$ may be referred to as the \textit{travel-times (TT) TSP polytope
(rooted at ``}$\mathit{0}$\textit{'')}, and its extreme points may be referred
to as \textit{travel-times (TT) TSP tours}. \smallskip
\end{enumerate}
\end{definition}

We will now provide a numerical illustration of the discussions above. \smallskip

\begin{example}
\label{Non-Applicability_Example} \ \medskip

Some of the differences between the \textit{travel-times TSP polytope} ($AP$)
and the \textit{travel-legs TSP polytope (}``\textit{The} TSP\ Polytope'')
discussed above will now be illustrated using a $6$-city TSP with node set $\{0,1,2,3,4,5\}.$

\begin{itemize}
\item Illustration of the\ \textit{travel-times TSP tours} on the $AP$
graphical tableau:
\begin{figure}[H]
\includegraphics[width=0.6\textwidth]{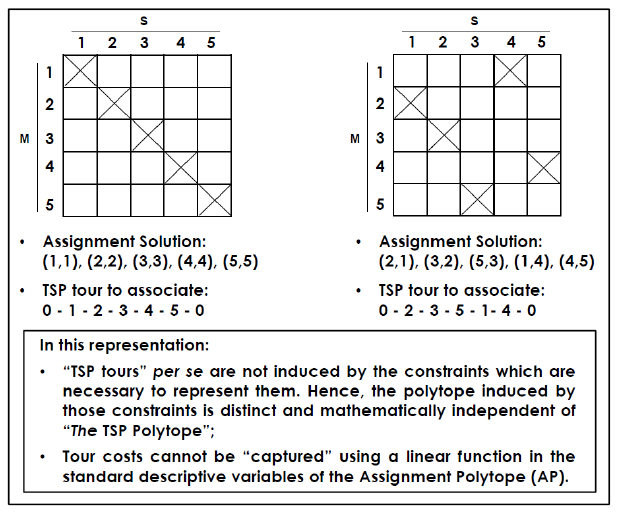}
\end{figure}
\item Illustration of the\ \textit{travel-times TSP tours} on the TSP Graph:
\begin{figure}[H]
\includegraphics[width=0.6\textwidth]{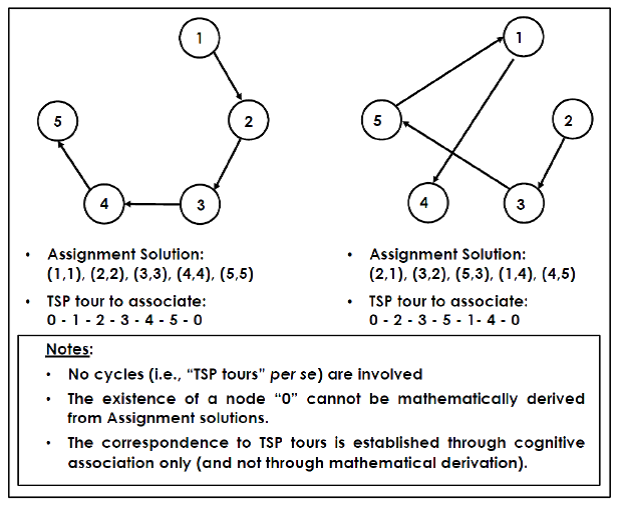}
\end{figure}
\item Illustration of the \textit{travel-legs TSP\ tours} on the TSP Graph:
\begin{figure}[H]
\includegraphics[width=0.6\textwidth]{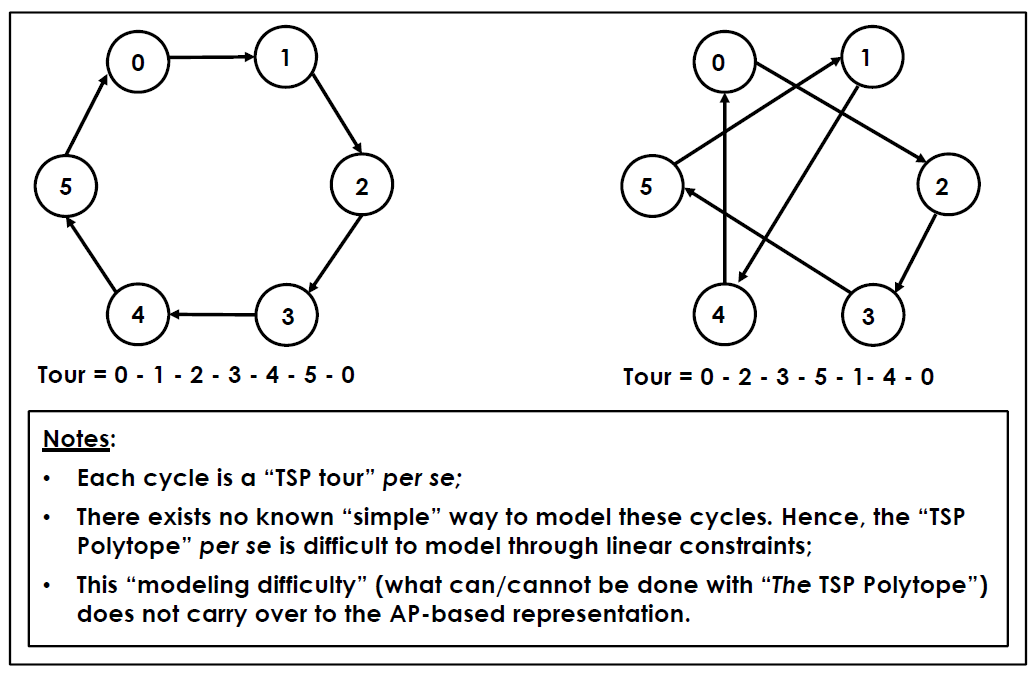}
\end{figure} 
\end{itemize}

\noindent$\square\medskip$
\end{example}

\subsection{The existence of an affine map does not necessarily imply
\textit{extension} relationships}
\label{Affine_Map_SubSection}

In the field of OR, it is common for two formulations of a problem expressed
in terms of disjoint sets of variables (say $x$ and $y$) to be equivalent. In
such a case, if an affine mapping between the sets of variables is known (say
from $y$ to $x$), one may add its expression to the model stated in the space
of the ``domain'' variables ($y$) and project the resulting \textit{augmented}
model onto the space of the ``range'' variables ($x$) for the purposes of
comparing the \textit{strengths} of the bounds that can be obtained from the
two formulations in a same space of variables (see Balas (2005, pp. 132-136).
The validity of this approach comes from the fact that the addition of
redundant variables and constraints to (i.e., an \textit{augmentation} of) an
optimization model does not change the objective function value of that model.
Note however, that a polytope cannot have a constraints description in the
space of variables which are not part of its set of descriptive variables.
Hence, as shown in Diaby and Karwan (2016; 2017), it is not valid to use the
projection of an \textit{augmentation} of a model in order to make inferences
about the \textit{size} of the constraints description of that model if the
vector space being projected to is disjoint from that of the model. Hence, as
has been mentioned earlier in this paper and also in Diaby and Karwan (2016;
2017), the belief/presumption that the existence of an affine map between
disjoint sets of variables describing different polytopes is sufficient to
imply $EF$ relationships from which valid inferences about model
\textit{sizes} can be made is a misconception. We will now demonstrate this by
deriving a sufficient condition for such an existence with no necessity nor
implication of \textit{extension} relations. We will also offer an alternate
interpretation of such an existence in an optimization context.

\begin{theorem}
\label{Affine_Mapping_Thm} \ \ \ \bigskip

\noindent Let:

\begin{description}
\item $\qquad$ $x\in\mathbb{R}^{p}$ and $y\in\mathbb{R}^{q}$ be disjoint
vectors of variables;

\item $\qquad$ $X:=\{x\in\mathbb{R}^{p}:Ax\leq a\}\neq\varnothing$
(where:$\ A\in\mathbb{R}^{k\times p};$ $a\in\mathbb{R}^{k}$);

\item $\qquad$ $L:=\left\{  \dbinom{x}{y}\in\mathbb{R}^{p+q}:Bx+Cy=b\right\}
\neq\varnothing$ (where: $B\in\mathbb{R}^{m\times p};$ $C\in\mathbb{R}
^{m\times q};$ $b\in\mathbb{R}^{m}$);

\item $\qquad$ $Y:=\{y\in\mathbb{R}^{q}:Dy\leq d\}\neq\varnothing$
(where:$\ D\in\mathbb{R}^{l\times q};$ $d\in\mathbb{R}^{l}$).
\end{description}

\noindent Then, the following are true:

\begin{enumerate}
\item \label{AMThm(1)}Provided $B^{T}B$ is nonsingular, there exists a
one-to-one affine function from $\mathbb{R}^{q}$ to $\mathbb{R}^{p}$ which
maps $y$ onto $x$.

\item \label{AMThm(2)}Assume the following conditions are true:

\begin{enumerate}
\item \label{AffCond(a)}$B^{T}B$ is nonsingular;

\item \label{AffCond(b)}The constraints of $L$ are \textit{redundant for} $X$
and $Y$ respectively;

\item \label{AffCond(c)}$L$ is the graph of a one-to-one correspondence
between the points of $X$ and the points of $Y$ (see Beachy and Blair (2006,
pp. 47-59)).
\end{enumerate}

Then, the optimization of any linear function of $x$ over $X$ can be done
without any reference to the constraints description of $X$.
\end{enumerate}
\end{theorem}

\begin{proof}
\ \ \smallskip

\begin{enumerate}
\item From the definition of $L$, we have:
\begin{equation}
Bx=b-Cy. \label{GenRefutProof(1)}
\end{equation}
Pre-multiplying ($\ref{GenRefutProof(1)}$) by $B^{T}$ gives:
\begin{equation}
B^{T}Bx=B^{T}b-B^{T}Cy. \label{genRefutProof(3)}
\end{equation}
Using the nonsingularity of $B^{T}B$ (according to the premise) and
pre-multiplying ($\ref{genRefutProof(3)}$) by $(B^{T}B)^{-1},$ we get:
\begin{equation}
x=(B^{T}B)^{-1}B^{T}b-(B^{T}B)^{-1}B^{T}Cy. \label{GenRefutProof(5)}
\end{equation}
Expression ($\ref{GenRefutProof(5)}$) can be written as:
\begin{align}
&  x=\overline{C}y+\overline{b},\text{ }\nonumber\\[0.09in]
&  \text{where: }\overline{C}:=-(B^{T}B)^{-1}B^{T}C\text{, and }\overline
{b}:=(B^{T}B)^{-1}B^{T}b. \label{GenRefutProof(7)}
\end{align}

\item Consider the task of minimizing the function $\alpha^{T}x$ ($\alpha
\in\mathbb{R}^{p}$) over $X$. This optimization problem can be expressed as:\smallskip

\textit{Problem LP}$_{0}$:\medskip\newline
\begin{tabular}
[c]{l}
\ \ \
\end{tabular}
$\left|
\begin{tabular}
[c]{ll}
$\text{Minimize:}$ & $\alpha^{T}x$\\
& \\
$\text{Subject To:}$ & $x\in X.$
\end{tabular}
\ \ \ \ \ \ \ \ \text{ \ }\right.  \bigskip$\newline Since the constraints of
$L$ are \textit{redundant} for $X$ and $Y$ respectively (according to premise
($b$)), \textit{LP}$_{0}$ is equivalent to:$\medskip$\newline \textit{Problem
LP}$_{1}$:\medskip\newline
\begin{tabular}
[c]{l}
\ \ \
\end{tabular}
$\left|
\begin{tabular}
[c]{ll}
$\text{Minimize:}$ & $\alpha^{T}x$\\
& \\
$\text{Subject To:}$ & $\dbinom{x}{y}\in L;$ $\ x\in X;$ \ $y\in Y.$
\end{tabular}
\ \ \ \ \ \ \ \ \text{ \ }\right.  \bigskip$\newline Using
(\ref{GenRefutProof(7)}) (since $B^{T}B$ is nonsingular according to premise
($a$)) to eliminate $x$ from the objective of \textit{Problem LP}$_{1};$ using
the fact that the constraints of $L$ are \textit{redundant }for $X$ and $Y$
respectively (according to premise ($b$)) to eliminate the constraints of $L$
from \textit{Problem LP}$_{1};$ and using the fact that $L$ is the graph of a
one-to-one correspondence between the points of $X$ and the points of $Y$
(according to premise ($c$)) to eliminate $X$ from the constraints set of
\textit{Problem LP}$_{1},$ gives that \textit{Problem LP}$_{1}$ can be solved
using the following two-step procedure:\medskip\newline \textit{Step 1}: Solve
\textit{Problem LP}$_{2}$: \medskip\ \newline
\begin{tabular}
[c]{l}
\ \ \ \ \ \ \
\end{tabular}
\ \ \ $\left|
\begin{tabular}
[c]{ll}
$\text{Minimize:}$ & $\left(  \alpha^{T}\overline{C}\right)  y+\alpha
^{T}\overline{b}$\\
& \\
$\text{Subject To:}$ & $y\in Y$
\end{tabular}
\ \ \ \ \ \ \ \ \text{ \ }\right.  \medskip\medskip$\newline
\begin{tabular}
[c]{l}
\ \ \ \
\end{tabular}
\ \ \ Let the solution be $y^{\ast}.$ \medskip

\textit{Step 2}: Use Graph $L$ to ``retrieve'' the optimal solution $x^{\ast}$
to \textit{Problem LP}$_{1}$:\bigskip\ \newline
\begin{tabular}
[c]{l}
\ \ \
\end{tabular}
\ \ \ \ \ $x^{\ast}=\overline{C}y^{\ast}+\overline{b}.$\ \medskip

Note that the second term of the objective function of \textit{Problem
LP}$_{2}$ can be ignored in the optimization process of \textit{Problem
LP}$_{2},$ since that term is a constant.\medskip

Also, if $L$ is derived from knowledge of the vertex (``$\mathcal{V}$-'')
representation of $X$ only (as would be the case if $X$ were ``\textit{The}
TSP Polytope'' for example), then this would mean that the $\mathcal{H}
$-representation of $X$ is not involved in the ``two-step'' solution process
above, but rather, that only the $\mathcal{V}$-representation of $X$ is involved.\ \ 
\end{enumerate}
\end{proof}

\bigskip Part ($\ref{AMThm(2)}$) of Theorem $\ref{Affine_Mapping_Thm}$ is
similar to Proposition 2 of Padberg and Sung (1991; p. 323).

\section{Specific refutations}
\label{Specific_Refutations_Section}

In all of the recent $EF$ developments with the claim that an
\textit{NP-Complete} COP cannot be abstracted into a \textit{polynomial-sized}
LP (i.e., Avis and Tiwary (2015), Bri\"{e}t \textit{et al}. (2015), Braun
\textit{et al}. (2015), Fiorini \textit{et al}. (2015), Lee \textit{et al}.
(2017), and Averkov \textit{et al}. (2018)), results are developed for the
natural polytopes of COPs using ``slack matrices''-based concepts (``extension
complexity;'' ``approximation complexity;'' ``relaxation complexity;'' etc.),
and their generalizations to all arbitrary abstractions of COPs hinges
(invariably) on the use of the notion that the existence of linear maps (which
are special-cases of affine maps) between disjoint sets of variables implies
$EF$ relationships between polytopes stated in the spaces of those sets of
variables. As discussed in Diaby and Karwan (2016; 2017), this notion is
\textit{degenerate} in the sense that it allows for every conceivable pair of
polytopes to be $EF$s of each other provided they are non-empty and are stated
in terms of disjoint sets of variables. Hence, the key to pinpointing the
inherent mathematical flaws in all of these recent $EF$ papers is to focus on
their generalization steps (which invariably need this \textit{degenerate}
$EF$ notion). We will illustrate this in this section using the Fiorini
\textit{et al}. (2015) developments, after providing a brief overview of the
background definitions which are used interchangeably in the ``line of research.''

\subsection{Background Definitions}

\begin{definition}[``Standard EF Definition'']
\label{EF_Dfn_Std}
An \textit{extended formulation} for a polytope $P=\left\{  x\in\mathbb{R}^{d}\right.  $ $\left.  :Ax\leq
b\right\}  \subseteq\mathbb{R}^{d}$ is a polyhedron $Q=\left\{  \dbinom{x}
{y}\in\mathbb{R}^{d+k}:Ex+Fy\leq g\right\}  ,$ the projection of which onto
$x$-space, $\varphi_{x}(Q):=\left\{  x\in\mathbb{R}^{d}:\left(  \exists
y\in\mathbb{R}^{k}:\dbinom{x}{y}\in Q\right)  \right\}  ,$ is equal to $P$
(where $A\in\mathbb{R}^{n\times d},$ $b\in\mathbb{R}^{d},\ E$ $\in
\mathbb{R}^{m\times d},$ $F\in\mathbb{R}^{m\times k},$ and $g\in\mathbb{R}
^{m}$) (Yannakakis (1991)).
\end{definition}

\begin{definition}[``Fiorini \textit{et al.} Definition \#1'']
\label{EF_Dfn_A1}
A polyhedron $Q$ $=$ $\left\{  \dbinom{x}{y}\in\mathbb{R}^{d+k}:Ex+Fy\right.  $ $\left.  \leq
g\right\}  $ is an \textit{extended formulation} of a polytope $P\subseteq
\mathbb{R}^{d}$ if there exists a linear map $\pi$ $:$ $\mathbb{R}
^{d+k}\longrightarrow\mathbb{R}^{d}$ such that $P$ is the image of $Q$ under
$\pi$ (i.e., $P=\pi(Q)$; where $E\in\mathbb{R}^{m\times d}$, $F\in
\mathbb{R}^{m\times k},$ and $g\in\mathbb{R}^{m}$) (Fiorini \textit{et al.}
(2015; p. 17:3, lines 20-21; p. 17:9, lines 22-23)).
\end{definition}

\begin{definition}[``Fiorini \textit{et al.} Definition \#2'']
\label{EF_Dfn_A2}
An \textit{extended formulation} of a polytope $P\subseteq\mathbb{R}^{d}$ is a
linear system $Q$ $=$ $\left\{  \dbinom{x}{y}\in\mathbb{R}^{d+k}:Ex+Fy\leq
g\right\}  $ such that $x\in P$ if and only if there exists $y\in
\mathbb{R}^{k}$ such that $\dbinom{x}{y}\in Q.$ (In other words, $Q$ is an
$EF$ of $P$ if $\left(  x\in P\Longleftrightarrow\left(  \exists
y\in\mathbb{R}^{k}:\dbinom{x}{y}\in Q\right)  \right)  $ (where $E$
$\in\mathbb{R}^{m\times d},$ $F\in\mathbb{R}^{m\times k},$ and $g\in
\mathbb{R}^{m}$) (Fiorini \textit{et al}. (2015; p. 17:2, last paragraph; p.
17:9, line 20-21)).$\medskip$
\end{definition}

\begin{remark}
\label{EF_Dfns_Rmk} \ \ 

\begin{enumerate}
\item \label{EFDefRmk(1)}
Because every equality constraint in an optimization
problem can be changed to a pair of inequalities, the expression of $Q$ in the
definitions above is general. However, the equality constraints are sometimes
seperated out in the $EF$ papers, and $E$, $F$, and $g$ are partitioned by
rows as $E^{=},$ $E^{\leq},$ $F^{=},$ $F^{\leq},$ $g^{=},$ and $g^{\leq},$
respectively, so that $Q$ is written as:
\begin{equation}
Q=\left\{  \dbinom{x}{y}\in\mathbb{R}^{d+k}:E^{\leq}x+F^{\leq}y\leq g^{\leq
};\text{ }E^{=}x+F^{=}y=g^{=}\right\}  \label{EF_Q_Dfn_Rmk1}
\end{equation}

(where $E^{\leq}\in\mathbb{R}^{m^{\leq}\times d},$ $E^{=}\in\mathbb{R}
^{m^{=}\times d},$ $F^{\leq}\in\mathbb{R}^{m^{\leq}\times k},$ $F^{=}
\in\mathbb{R}^{m^{=}\times k},$ $g\in\mathbb{R}^{m^{\leq}},$ and
$g\in\mathbb{R}^{m^{=}},$ with $m^{=}+m^{\leq}=m$). \ 

\item \label{EFDefRmk(2)}Definition \ref{EF_Dfn_Std} is the standard,
reference $EF$ definition.

\item \label{EFDefRmk(3)}Definitions \ref{EF_Dfn_A1} and \ref{EF_Dfn_A2} are
alternate $EF$ definitions which are used in Fiorini \textit{et al}. (2015)
and other recent $EF$ papers.

\item \label{EFDefRmk(4)}Definition \ref{EF_Dfn_A2} is consistent with
Definition \ref{EF_Dfn_Std}.

\item \label{EFDefRmk(5)}Definition \ref{EF_Dfn_A1} is inconsistent with
Definition \ref{EF_Dfn_Std} (and therefore, with Definition \ref{EF_Dfn_A2}
also), when the description of $Q$ does not involve the $x$-variables (i.e.,
when $E=\mathbf{0}$; see Diaby and Karwan (2016; 2017)).

\item \label{EFDefRmk(6)}The use of Definition \ref{EF_Dfn_A1} in the recent
$EF$ papers is what allows them to generalize their results beyond the natural
polytopes of COPs. However, this use of Definition \ref{EF_Dfn_A1} is also
what makes the Mathematics in those $EF$ papers inherently flawed. This will
be illustrated in the remainder of this section. \ \ 
\end{enumerate}

\noindent$\square\medskip$
\end{remark}

\subsection{Inherently-flawed Mathematics: Illustration using Fiorini
\textit{et al.} (2015)}

\begin{example}
\label{No_EF-Relation_Example}
: Let $\mathbf{x}\in\mathbb{R}^{3}$ and
$\mathbf{y\in}\mathbb{R}$ be disjoint vectors of variables. Let $P$ be a
polytope in the space of $\mathbf{x}$, and $Q,$ a polytope in the space of
$\dbinom{\mathbf{x}}{\mathbf{y}}$, with:
\begin{align}
P  &  :=Conv\left(  \left\{  \left(
\begin{array}
[c]{c}
8\\
10\\
6
\end{array}
\right)  ,\left(
\begin{array}
[c]{c}
12\\
15\\
9
\end{array}
\right)  \right\}  \right)  \text{, and}\label{EF_NumEx(a)}\\
Q  &  :=\left\{  \dbinom{\mathbf{x}}{\mathbf{y}}\in\mathbb{R}^{3+1}
:2\leq\mathbf{0}\cdot\mathbf{x}+\mathbf{y}\leq3\right\}  . \label{EF_NumEx(b)}
\end{align}

We will now discuss some key results of Fiorini \textit{et al}. (2015) which
are refuted by $P$ and $Q.$

\begin{enumerate}
\item \textit{Refutation of the validity of Definition \ref{EF_Dfn_A1}.}

\begin{enumerate}
\item Note that the following is true for $P$ and $Q$:
\begin{equation}
\mathbf{x}\in P\nLeftrightarrow\left(  \exists\mathbf{y\in}\mathbb{R}
:\dbinom{\mathbf{x}}{\mathbf{y}}\in Q\right)  . \label{EF_NumEx(d)}
\end{equation}
For example,
\begin{equation}
\left(  \exists\mathbf{y\in}\mathbb{R}:\left(
\begin{array}
[c]{c}
22.5\\
-50\\
100\\
y
\end{array}
\right)  \in Q\right)  \nRightarrow\left(  \left(
\begin{array}
[c]{c}
22.5\\
-50\\
100
\end{array}
\right)  \in P\right)  . \label{EF_NumEx(e)}
\end{equation}
\ \ Hence, $Q$ \textbf{is not} an \textit{extended formulation} of $P$
according to Definition \ref{EF_Dfn_A2}.

\item Observe that the following is also true for $P$ and $Q$:
\begin{align}
P\text{ }=  &  \text{ }\left\{  \mathbf{x\in}\mathbb{R}^{3}:\mathbf{x}
=\pi\cdot\dbinom{\mathbf{z}}{\mathbf{y}}\mathbf{,}\text{ }\mathbf{z}
\in\mathbb{R}^{3},\text{ }\dbinom{\mathbf{z}}{\mathbf{y}}\in Q\right\}
\label{EF_NumEx(f1)}\\[0.09in]
\text{ }=  &  \text{ \ }\pi(Q);
\label{EF_NumEx(f2)}\\[0.09in]
\text{wh}\text{e}\text{r}  &  \text{e }\pi\text{ }\mathbf{=}\text{ }\left[
\begin{array}
[c]{cccc}
0 & 0 & 0 & 4\\
0 & 0 & 0 & 5\\
0 & 0 & 0 & 3
\end{array}
\right]  \text{.} 
\label{EF_NumEx(f3)}
\end{align}
In other words, $P$ is the image of $Q$ under the linear map $\pi$. Hence, $Q$
\textbf{is} an\textit{ extended formulation} of $P$ according to Definition
\ref{EF_Dfn_A1}.

\item It follows from (a) and (b) above, that Definitions \ref{EF_Dfn_A1} and
\ref{EF_Dfn_A2} are in contradiction of each other with respect to $P$ and $Q$.

Hence $P$ and $Q$ are a refutation of the validity of Definition
\ref{EF_Dfn_A1}, since Definition \ref{EF_Dfn_A2} is equivalent to Definition
\ref{EF_Dfn_Std}, which is the standard/reference definition (as indicated in
Remarks \ref{EF_Dfns_Rmk}.\ref{EFDefRmk(2)}\ and \ref{EF_Dfns_Rmk}
.\ref{EFDefRmk(4)}).
\end{enumerate}

\item \textit{Refutation of ``Theorem 3''} (p.17:10) of Fiorini \textit{et
al}. (2015).

The proof of the theorem (``Theorem 3'') hinges on Definition \ref{EF_Dfn_A1}.
The specific statement in Fiorini \textit{et al}. (2015; p. 17:10, lines
26-28) is:
\begin{align}
&  \text{``}...\text{\textit{Because} }\nonumber\\[0.06in]
&  \mathit{Ax\leq b\Longleftrightarrow\exists y:E}^{=}\mathit{x+F}
^{=}\mathit{y=g}^{=}\mathit{,}\text{ }\mathit{E}^{\leq}\mathit{x+F}^{\leq
}\mathit{y\leq g}^{=}\mathit{,}\label{EF_NumEx(g)}\\
&  \text{\textit{each inequality in} }\mathit{Ax\leq b}\text{\textit{is valid
for all points of} }\mathit{Q}\text{. ...''}\nonumber
\end{align}

The equivalent of (\ref{EF_NumEx(g)}) in terms of $P$ and $Q$ (using the
partitioned form (\ref{EF_Q_Dfn_Rmk1}) for $Q$) is:
\begin{equation}
\mathbf{x}\in P\Longleftrightarrow\exists\mathbf{y\in}\mathbb{R}
:\dbinom{\mathbf{x}}{\mathbf{y}}\in Q. \label{EF_NumEx(h)}
\end{equation}
Clearly, (\ref{EF_NumEx(h)}) is \textbf{not true}, since it is in
contradiction of expressions (\ref{EF_NumEx(d)})-(\ref{EF_NumEx(e)}) above.
Hence, the proof of ``Theorem 3'' (and therefore, ``Theorem 3'') of Fiorini
\textit{et al}. (2015) is refuted by $P$ and $Q$.

\item \textit{Refutation of ``Lemma 9''} (p. 17:13-17:14) of Fiorini
\textit{et al}. (2015)\textit{.}

The first part of the lemma is stated (in Fiorini \textit{et al}. (2015))
thus:
\begin{align*}
&  \text{``\textit{Lemma 9. Let }P\textit{, }Q\textit{, and }F\textit{\ be
polytopes. Then, the following hold:}}\\
&  \text{\textit{(i) if F is an extension of P, then }}xc\text{(F)}\geq
xc\text{(P); \ldots''\ }
\end{align*}

The proof of this (in Fiorini \textit{et al}. (2015)) is stated as follows:
\[
\text{``\textit{Proof. The first part is obvious because every extension of
F\ is in particular an extension of P. }\ldots''}
\]

The notation ``$xc(\cdot)"$ in these statements (of Fiorini \textit{et al}.
(2015)) stands for ``\textit{extension complexity} of ($\cdot$),'' which is
defined as (p. 17:9, lines 24-25 of Fiorini \textit{et al}. (2015)):
\begin{align*}
&  \text{``...\textit{the extension complexity of P is the minimum size (i.e.,
the number of inequalities) of an EF }}\\
&  \text{\textit{of P}.''}
\end{align*}

The refutation of the Fiorini \textit{et al}. (2015) \textit{``Lemma 9''} for
$P$ (as shown in (\ref{EF_NumEx(a)}) above) and $Q$ (as shown in
(\ref{EF_NumEx(b)}) above) is as follows.

First, note that (\ref{EF_NumEx(b)}) can re-written in its explicit form as:
\begin{equation}
Q=\left\{  (\mathbf{x,y)\in(}\mathbb{R}^{3},\mathbb{R)}:2\leq y\leq3\right\}
.\label{Counter_Example_Description_1}
\end{equation}

As shown in Part ($1$) above (in this paper), $Q$ is an \textit{extension} of
$P$ according to Definition \ref{EF_Dfn_A1} (which is central in Fiorini
\textit{et al}. (2015)). Accordingly, therefore, this means that $Q$ is an
\textit{extended formulation} of every one of the infinitely-many possible
$\mathcal{H}$-descriptions of $P$. This would be true in particular for the
$\mathcal{H}$-description below for $P$:
\begin{equation}
P=\left\{
\begin{array}
[c]{l}
\mathbf{x\in}\mathbb{R}^{3}:\\
\\
-5x_{1}+4x_{2}\leq0;\\
\text{ }\\
3x_{2}-5x_{3}=0;\text{ }\\
\\
3x_{1}-4x_{3}\leq0;\\
\\
8\leq x_{1}\leq12;\text{ }\\
\\
10\leq x_{2}\leq15;\\
\text{ }\\
6\leq x_{3}\leq9
\end{array}
\right\}  .\label{Counter_Example_Description_2}
\end{equation}
Clearly, the number of inequalities in (\ref{Counter_Example_Description_2})
is greater than the number of inequalites in
(\ref{Counter_Example_Description_1}).

In other words, for $P$ and $Q$, we have that:
\begin{equation}
xc(Q)\ngeq xc(P). \label{EF_NumEx(i)}
\end{equation}
Hence, $P$ and $Q$ are a refutation of ``Lemma 9'' of Fiorini \textit{et al}.
(2015), being that $Q$ is the \textit{extension}, and\textit{\ }$P$, the
\textit{projection,} according to the definitions used in Fiorini \textit{et
al}. (2015).
\end{enumerate}

\noindent$\square\medskip$
\end{example}

\begin{remark}
\ \ 

\begin{enumerate}
\item According to Fiorini \textit{et al}. (2015; p. 17:7, Section 1.4, first
sentence; p. 17:11, lines 6-11; p. 17:14, lines 5-6; p.17:16, lines 13-14
after the ``Fig. 4'' caption), their ``Theorem 3'' and ``Lemma 9'' play
pivotal, foundational roles in the rest of their developments. Note that
``Lemma 9'' (of Fiorini \textit{et al.} (2015)) does not depend on any one of
the \textit{extended formulations} definitions used in Fiorini \textit{et al.}
(2015) in particular. Hence, we believe the numerical illustration we have
provided above represents a simple-yet-complete refutation of the developments
in Fiorini \textit{et al.} (2015).

\item A ``feature'' of $Q$ in our counter-example above is that its minimal
inequality (``outer'' or ``$\mathcal{H}$-'') description does not require the
$x$-variables. Hence, $P$ and $Q$ in the example essentially have disjoint
sets of descriptive variables. Hence, as shown in Diaby and Karwan (2016;
2017), the $EF$s relationship which would be created between the two polytopes
by the addition of the expression of the linear map in ($\ref{EF_NumEx(f1)}
$)-($\ref{EF_NumEx(f3)}$) to the description of $Q$ is \textit{degenerate}
/meaningless with respect to the task of making valid inferences about the
\textit{size} of the $\mathcal{H}$-description of $Q$ from the \textit{size}
of the $\mathcal{H}$-description of $P$. The reason for this is that the
derivation of the linear map involves the extreme-point (``inner-'' or
``$\mathcal{V}$-'') description of $P$ only, as detailed in Diaby and Karwan
(2016; 2017). \ \ 
\end{enumerate}

\noindent$\square$
\end{remark}

\section{Conclusions}
\label{Conclusion_Section}

In this paper, we have provided a multi-leveled refutation of the claim in
some of the recent \textit{extended formulations} ($EF$s) papers (Avis and
Tiwary (2015), Bri\"{e}t \textit{et al}. (2015), Braun \textit{et al}. (2015),
Fiorini \textit{et al}. (2015), Lee \textit{et al}. (2017), and Averkov
\textit{et al}. (2018)) that an \textit{NP-Complete} problem cannot be
abstracted into a \textit{polynomial-sized} linear program (LP). One of the
two fundamental misconceptions in those papers is the (sometimes-implicit)
belief that all abstractions of a combinatorial optimization problem (COP)
\textit{must} involve the polytope stated in terms of the natural variables
for that COP (for example, ``\textit{Th}e TSP Polytope'' for the TSP). We have
provided a direct refutation of this misconception by exhibiting a
\textit{polynomial-size} LP model which we have shown to correctly abstract
TSP tours. The other misconception in the ``unconditional impossibility''
papers is the belief that the existence of an affine map between disjoint sets
of variables implies $EF$ relationships from which valid inferences about
model \textit{sizes} can be made. In order to refute this misconception on a
general level, we have developed conditions for such existence independently
of any implication or need for an \textit{extension} relationship. We have
also clarified the meaning of the existence of such a mapping in an
optimization context.

A consequence of the presumption that the existence of an affine map between
disjoint sets of variables implies meaningful/non-degenrate $EF$ relationships
between polytopes in the spaces of those sets of variables is that it
introduces inherent flaws in the Mathematics of the papers that use it in
their efforts to make inferences about model \textit{sizes}. Focusing on this,
we have provided counter-example refutations of the key foundational results
of Fiorini \textit{et al}. (2015) (namely, their ``Theorem 3'' and ``Lemma 9''
) upon which (according to them) their claims that:\smallskip

\begin{quotation}
``\textit{We solve this question by proving a super-polynomial bound on the
number of }\newline \textit{inequalities in every LP for the TSP.}'' (Fiorini
\textit{et al}. (2015, p.17:2, lines 9-10)); \medskip

``\textit{We} \textit{also} \textit{prove such unconditional super-polynomial
bounds for the maximum cut }\newline \textit{and} \textit{the}
\textit{maximum} \textit{stable} \textit{set} \textit{problems.''} (Fiorini
\textit{et al}. (2015, p.17:2, lines 10-12));
\end{quotation}

and

\begin{quotation}
``...\textit{it is} \textit{impossible to prove} \textit{P = NP by means}
\textit{of a polynomial-sized LP that }\newline \textit{e\textit{x}presses any
of these problems.}'' (Fiorini \textit{et al}. (2015, p.17:2, lines 12-13))\smallskip
\end{quotation}

\noindent are based. In other words, we have shown that these claims of
Fiorini \textit{et al}. (2015) are not supported by the Mathematics in their
paper. The approach we have used in order to do this can be readily applied to
all of the other papers in this ``line of research,'' including Bri\"{e}t
\textit{et al}. (2015), Braun \textit{et al}. (2015), Lee \textit{et al}.
(2017), and Averkov \textit{et al}. (2018)).

As far as we know, very few papers have been published, which offer alternate
abstractions of \textit{NP-Complete }COPs which do not require their natural
variables. Some exceptions are Diaby (2007; 2010a; 2010b; 2010c), Maknickas
(2015), and Diaby and Karwan (2016). Our suggestion for future directions is
that research focus be shifted away from developments aimed at showing
negative results for the natural polytopes of COPs in favor of efforts
directed at the exploration of novel approaches which may yield
low-dimensional alternate abstractions instead.\pagebreak

\end{document}